\documentclass[10pt]{article}

\usepackage[letterpaper,margin=1in]{geometry}

\usepackage[sort]{natbib}
\usepackage{amsmath}
\usepackage{amsthm}
\usepackage{shortcutsg}
\usepackage{bm}
\usepackage{tikz}
\usepackage{subcaption}
\usepackage[capitalise]{cleveref}
\Crefname{assumption}{Assumption}{Assumptions}
\usetikzlibrary{arrows.meta,positioning,arrows,decorations}

\theoremstyle{plain}
\newtheorem{theorem}{Theorem}
\newtheorem{lemma}[theorem]{Lemma}

\theoremstyle{definition}
\newtheorem{assumption}{Assumption}

\newtheorem{definition}{Definition}

\newcommand{\pr}{\mathrm{Pr}}

\usepackage{xcolor}

\title{Long-Term Causal Inference with Imperfect Surrogates using Many Weak Experiments, Proxies, and Cross-Fold Moments}

\author{Aur\'elien Bibaut \and Nathan Kallus \and Simon Ejdemyr \and Michael Zhao}
\date{}

\begin{document}
\maketitle

\begin{abstract}\vspace{-2.5em}\noindent\textbf{Summary}\quad
Inferring causal effects on long-term outcomes using short-term surrogates is crucial to rapid innovation. However, even when treatments are randomized and surrogates fully mediate their effect on outcomes, it's possible that we get the direction of causal effects wrong due to confounding between surrogates and outcomes -- a situation famously known as the surrogate paradox. The availability of many historical experiments offer the opportunity to instrument for the surrogate and bypass this confounding. However, even as the number of experiments grows, two-stage least squares has non-vanishing bias if each experiment has a bounded size, and this bias is exacerbated when most experiments barely move metrics, as occurs in practice. We show how to eliminate this bias using cross-fold procedures, JIVE being one example, and construct valid confidence intervals for the long-term effect in new experiments where long-term outcome has not yet been observed. Our methodology further allows to proxy for effects not perfectly mediated by the surrogates, allowing us to handle both confounding and effect leakage as violations of standard statistical surrogacy conditions. We further analyze regularized cross-fold procedures for the high-dimensional setting with rich surrogates, as would be needed in practice to justify exclusion restriction.
\end{abstract}

\section{Long-Term Causal Inference Using Surrogates}

\begin{figure}[b]
\centering
\caption{Causal diagrams for surrogate settings with unobserved confounders. Dashed circles indicate unobserved variables. Dotted circles indicate variables observed historically, but unobserved for novel treatments.}
\begin{subfigure}[m]{0.45\textwidth}
\centering
\begin{tikzpicture}
\node[draw, circle, text centered, minimum size=0.75cm, line width= 1] (a) {$A$};
\node[draw, circle, right=1 of a, text centered, minimum size=0.75cm, line width= 1] (s) {$S$};
\node[draw, circle, above right=0.5 and 0.325 of s,text centered, minimum size=0.75cm, dashed,line width= 1] (u) {$U$};
\node[draw, circle, right=1 of s, text centered, minimum size=0.75cm, dotted, line width= 1] (y) {$Y$};
\draw[-latex, line width= 1] (a) -- (s);
\draw[-latex, line width= 1] (s) -- (y);
\draw[-latex, line width= 1] (u) -- (s);
\draw[-latex, line width= 1] (u) -- (y);
\end{tikzpicture}
\caption{A setting with unconfounded treatment ($A$) but confounded surrogate ($S$) and outcome ($Y$).}\label{fig:simple}
\end{subfigure}
\hfill
\begin{subfigure}[m]{0.525\textwidth}
\centering
\begin{tikzpicture}
\node[draw, circle, text centered, minimum size=0.75cm, line width= 1] (a) {$A$};
\node[draw, circle, above right=0 and 1 of a, text centered, minimum size=0.75cm, line width= 1] (s) {$S_1$};
\node[draw, circle, above right=0.25 and 0.325 of s,text centered, minimum size=0.75cm, dashed,line width= 1] (u) {$U_1$};
\node[draw, circle, below right=0 and 1 of a,text centered, minimum size=0.75cm, dashed,line width= 1] (u2) {$U_2$};
\node[draw, circle, right=1 of s, text centered, minimum size=0.75cm,dotted, line width= 1] (y) {$Y$};
\node[draw, circle, right=1 of u2, text centered, minimum size=0.75cm, line width= 1] (s2) {$S_2$};
\draw[-latex, line width= 1] (a) -- (s);
\draw[-latex, line width= 1] (s) -- (y);
\draw[-latex, line width= 1] (u) -- (s);
\draw[-latex, line width= 1] (u) -- (y);
\draw[-latex, line width= 1] (a) -- (u2);
\draw[-latex, line width= 1] (u2) -- (s2);
\draw[-latex, line width= 1] (u2) -- (s);
\draw[-latex, line width= 1] (u2) -- (y);
\draw[-latex, line width= 1] (s2) -- (y);
\end{tikzpicture}
\caption{The general setting we tackle with two kinds of possible short-term observations $S=(S_1,S_2)$. The identity of the components $S_1$ and $S_2$ within $S$ need not be known.}\label{fig:complex}
\end{subfigure}
\end{figure}
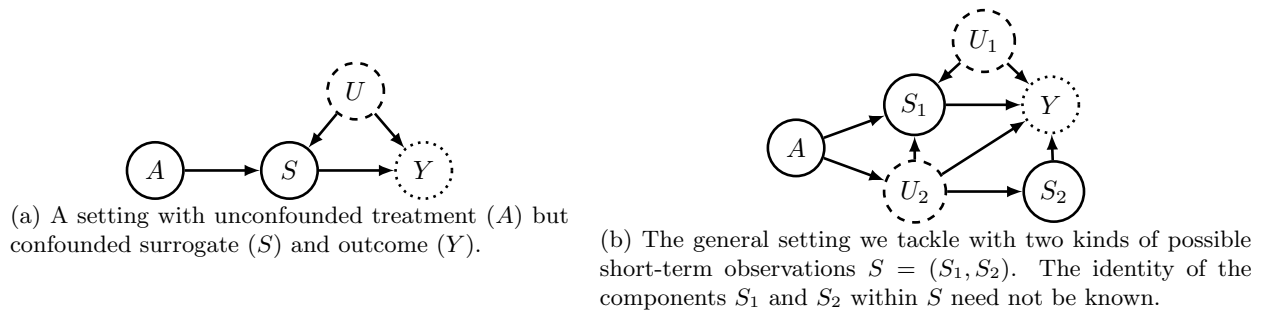

The long-term effect of interventions is often of primary concern in causal inference. Examples include the effect of early-childhood education on lifetime earnings \citep{chetty2011does}, of promotions on long-term value \citep{yang2020targeting}, or of digital platform design on long-term user retention \citep{hohnhold2015focusing}. While the gold standard for causal inference is experimentation, the significant delay of long-term observations after assignment to treatment means that, even when we can randomize the intervention, we may not be able to measure the outcome. Nevertheless, other relevant post-treatment observations are often available in the short-term. For example, in AIDS treatments we observe short-term viral loads or CD4 counts well before mortality outcomes \citep{fleming1994surrogate}. Similarly, in digital experimentation, we observe short-term signals on engagement well before retention or revenue shifts.

Leveraging these short-term observations to assist in inferring long-term effects has generated immense interest in a variety of settings. \citet{athey2020combining} and \citet{imbens2022long} consider the use of short-term observations from randomized experiments to remove confounding in observational studies of long-term outcomes. \citet{kallus2020role} consider the efficiency gains from including units with only short-term observations in analyses of experiments with long-term observations for some but not all units.

Possibly the most common setting assumes short-term observations form a \emph{statistical surrogate} \citep{prentice1989surrogate}. This requires the long-term outcome to be conditionally independent of treatment given the short-term outcomes. As this assumption may become more defensible as we include more short-term outcomes so as to mediate more of the treatment's long-term effect, \citet{athey2019surrogate} combine many short-term observations into a surrogate index, assuming they form a statistical surrogate and using historical data to regress long-term on short-term (or other ways of adjustment such as weighting). The proposal is simple, effective, and, as such, widely adopted.

\subsection{The Surrogate Paradox}

However, even if short-term outcomes fully mediate the long-term effect, they may fail to satisfy statistical surrogacy. Consider the causal diagram in \cref{fig:simple}, where $S$ perfectly mediates $A$'s effect on $Y$ (\ie, exclusion restriction), but $S$ and $Y$ share an unobserved confounder $U$, while treatment $A$ is fully unconfounded. In this case, $S$ is a collider so that conditioning on it induces a path from $A$ to $Y$ via $U$, violating surrogacy and imperiling analysis using such methods as \citet{athey2019surrogate}. 
% This scenario is actually exceedingly commonplace: for example, even if the effect of primary-school interventions on academic performance in university is mediated by academic performance in primary and secondary school, academic performance across all levels is reasonably jointly affected by inherent ability, intelligence, etc.

% This subtlety can lead to counterintuitive results. For example, making primary- and secondary-school curricula easier would have a positive effect on primary- and secondary-school performance, but cause declines in university level performance. This is the so-called surrogate paradox.
This scenario is actually exceedingly commonplace: for example, in streaming platforms with a subscription model, a user's amount of free time both impacts their short-term engagement and their probability of subscription retention in the same direction. Failing to address this confounding leads to surrogate-based estimates that overestimate true effects on long-term outcomes. We might even have more extreme situations in which an intervention strongly increases short-term engagement of a subpopulation of users who are unlikely to unsubscribe while slightly decreasing engagement of a subpopulation of users very likely to unsubscribe. This might result in both an overall increase in short-term engagement and an overall decrease in long-term retention, a situation known as the surrogate paradox.

\subsection{Experiments as Instruments and the Bias of 2SLS}

\citet{athey2019surrogate} consider a setting where historical data prior to the present experiment contains only $S$ and $Y$ and some baseline covariates. In this context, all we can do is either worry about potential unobserved confounders, or hope the included covariates satisfy ignorability between $S$ and $Y$. However, at organizations that routinely run many digital experiments, we can take historical data from past experiments where we also observe the randomized treatments $A$. In the setting of \cref{fig:simple}, these treatments constitute an instrumental variable, which can help us identify the causal effect of $S$ on $Y$ and therefore infer the effect of a novel treatment on $Y$ by considering only its effect on $S$. (The precise identification conditions for this are covered in the next section in an even more general setting.)

Leveraging multiple experiments as instruments was considered by \citet{peysakhovich2018learning}. A challenge with estimation highlighted in the paper is asymptotic consistency. This is because the relevant asymptotic regime is \emph{not} the sample size per experiment growing large, but rather the number of experiments growing large with a bounded sample size per experiment. 
\citet{peysakhovich2018learning} show that in this regime, 2-stage least squares (2SLS) has a non-zero asymptotic bias. Intuitively, this phenomenon can be understood to arise due to correlated errors in effect estimates across experiments. In this setting, 2SLS is equivalent to regressing experiment-cell means of $Y$ on experiment-cell means of $S$. Since cell sizes are bounded, each experiment-cell mean has non-vanishing error. Also, we try to choose $S$ such that it is closely related to $Y$, so these errors are also correlated, leading to a non-vanishing error-in-variables bias in the regression of experiment-cell means, and equivalently, in 2SLS.  An important contribution of this paper is that we can avoid this bias via cross-fold estimators (such as JIVE \citep{angrist1999jackknife}), in addition to the more general identification and regularization results.

\section{Fixing Surrogacy Violations via Instrumentation and Proxies}

We consider a setting where we have a historical data set consisting of   unit-level observations $(A, S, Y) \sim \mathbb{P}$, where $A$ is a unit's test cell assignment, $S$ is a vector of short-term metrics (\eg, measured a week or two after cell assignment), and $Y$ is a long-term outcome (\eg, measured multiple months after cell assignment). We provide conditions under which we can identify the mean potential outcome $E[Y(a')]$ under a previously completely unseen treatment $a'$ using the distribution of $S$ under $a'$. The latter can of course be gleaned from measuring only short-term outcomes in an experiment run with $a'$.

We suppose that the short-term vector $S$ can be partitioned into two types of measurements $S_1$ and $S_2$, where $S_1$ mediates (part of) the effect of $A$ on $Y$, and $S_2$ is a proxy measurement of an unmeasured vector $U_2$ that accounts for the effect of $A$ on $Y$ that is not mediated by $S_1$. We allow for various violations of Prentice's statistical surrogacy condition, which is the cornerstone of previous approaches. First, we allow for the surrogate $S_1$ not to fully mediate the effect of $A$ on $Y$, as long as there is a rich enough proxy $S_2$ for the unmediated part of the effect. Second, we allow for unmeasured confounding between the surrogate $S_1$ and the long-term outcome $Y$. We formally encode these two aspects in the following two assumptions.

\begin{assumption}[Existence of potential outcomes]\label{asm:potential_outcomes}
    There exist random variables $(S_1(a), U_2(a), S_2(u_2), Y(s_1, s_2, u_2) : a, s_1, s_2, u_2)$, which we refer to as \textit{potential outcomes} such that $S_1 = S_1(A)$, $S_2 = S_2(U_2)$, and $Y = Y(S_1, S_2, U_2)$, where we define $U_2 = U_2(A)$.
\end{assumption}
Under assumption \ref{asm:potential_outcomes}, we can define two additional potential outcomes as follows.
\begin{definition}\label{def:potential_outcomes}
    Let $S_2(a) = S_2(U_2(a))$ and $Y(a) = Y(S_1(a), S_2(a), U_2(a))$.
\end{definition}

\begin{assumption}[Unmeasured confounder and potential outcomes independence]\label{asm:indep_pot_outcomes}
    There exists a random variable $U_1$ such that 
    \begin{align}
         (U_1, Y(s_1, s_2, u_2), S_1(a)) &\indep (U_2(a), S_2(u_2)) \\
         Y(s_1, s_2, u_2) & \indep S_1(a) \mid U_1.
    \end{align}
    Furthermore, $U_2(a) \indep S_2(u_2)$.
\end{assumption}

For identification of causal effects to be possible in spite of the unmeasured confounding between $S_1$ and $Y$, we require that cell assignment be randomized so that it can instrument for $S_1$. 
We encode this formally in the next three assumptions.

\begin{assumption}[Randomization]\label{asm:randomization}
    It holds that $$A \indep (U_1, U_2(a), S_1(a), S_2(u_2), Y(s_1, s_2, u_2) : a, s_1, s_2, u_2);$$
\end{assumption}
See fig \ref{fig:complex} for a causal graph compatible with  assumptions \ref{asm:potential_outcomes}-\ref{asm:randomization}.

For the instrumental variable to allow for identification of an ATE, we now make an assumption akin to requiring the effect of the surrogate $S_1$ on $Y$ to be homogeneous across principal strata. The exact assumption is actually slightly more general than no heterogeneity in principal strata and goes as follows.
\begin{assumption}[No treatment effect heterogeneity across principal strata]\label{asm:homogeneous_treatment_effect}
Letting $\varphi(s_1, s_2, u_1, u_2) = E[Y(s_1, s_2, u_2) \mid U_1=u_1] - E[Y(s_1, s_2, u_2)]$, it holds that $E[\varphi(S_1(a), s_2, U_1, u_2)] = 0$ for every $a$, $s_2$, $u_2$.
\end{assumption}
 Note that the absence of heterogeneity of effect across principal strata is a common condition implied by both parametric structural models and the nonparametric IV model \citep{newey2003instrumental}.

Finally, we impose two assumptions that are relatively common in the proximal inference literature, namely that the existence of a bridge function and a completeness condition.
\begin{assumption}[Bridge function existence]\label{asm:bridge_existence}
    There exists $h$ such that $E[Y - h(S_1, S_2) \mid A] = 0$ almost surely.
\end{assumption}

\begin{assumption}[Completeness]\label{asm:completeness}
    For every function $g$, ($E[g(S_1, U_2) \mid A]$ = 0 almost surely) implies $(g=0)$.
\end{assumption}

\begin{theorem}\label{thm:identification} In the setting of \cref{fig:complex} and under the above assumptions,
if $h$ solves the conditional moment restriction $E[Y - h(S) \mid A] = 0$ under the distribution of the historical dataset, then $E[Y(a')] = E[h(S(a'))]$ for a new intervention $a'$.
\end{theorem}
That is, from the historical data, we can learn a function $h$ that \textit{bridges} the effect of $A$ on $S$ to that on $Y$ so that in new experiments $h(S)$ can serve as a \textit{surrogate outcome}, which is based only on short-term observations and the average effect on which coincides with that on the unobserved long-term outcome. Note we do not actually need to know which part of $S$ belongs to $S_1$ or to $S_2$.

\section{Cross-Fold Moments for Learning the Bridge Function from Many Weak Experiments}

As we mentioned earlier, it is commonplace that even at the scale of experiments run by digital platforms, the effects of experiments on any metric $S$ predictive of long-term outcomes is very small, rendering cell assignments $A$ weak instruments for $S$. This motivates us to work in the asymptotic regime where the number of experimental cells $K$ grows to $\infty$ while both the number of units per cell stays bounded and cells' effect on $S$ is bounded or even diminishing. The bias of usual instrumental variable estimator under the many-weak-experiments regime pointed out by \cite{peysakhovich2018learning} carries over to our more general setting and imperils standard two-stage and method-of-moments estimators for $h$. 

\subsection{Cross-fold risk}

In linear structural models, a solution to the many-weak instruments-asymptotics problem is the leave-one-out cross-fold estimator JIVE. We propose a generalized risk based on a cell-splitting device that allows for asymptotically unbiased (in the many-weak-experiments regime) estimation of bridge functions in general function classes such as neural networks or tree-based classes. Specifically, enriching observation triples $(A, S, Y)$ with an endogenously randomized fold assignment $V\sim\mathrm{Unif}\{1,\dots,L\}$, and letting $\widetilde{V}$ be an independent copy of $V$, we define the cross-fold risk for the bridge function $h$ as 
\begin{align}
    R(h) = -E[E[Y \mid A, V = \widetilde{V}] E[h(S) \mid A, V \neq \widetilde{V}]] + \frac{1}{2}  E[E[h(S) \mid A, V = \widetilde{V}]E[h(S) \mid A, V \neq \widetilde{V}]].\label{eq:cross-fold_risk}
\end{align}
 The cross-fold device allows to capture only the co-variation of $Y$ and $h(S)$ that is due to the experimental cell assignment. In particular, in linear structural models, correlated additive errors terms to $S$ and $Y$ play no role in $R(h)$, and do not appear in any minimizer of the risk, unlike is the case in 2SLS. Moreover, in a linear error-in-variables model, minimizers of the cross-fold risk are immune to attenuation bias. 

It is straightforward to check that measurable minimizers of the cross-fold risk solve the conditional moment restriction. We now give consistency guarantees as $K \to \infty$ under metric entropy conditions in general nonparametric function-approximating settings.

\subsection{Consistency of empirical cross-fold risk minimizers over nonparametric classes}

Let $\widehat h$ be a minimizer over $\mathcal{H}$ of the empirical counterpart of $R$.

Let $\Pi:\mathcal{L}_2(S) \to \mathcal{L}_2(A)$ defined, for every $h \in \mathcal{L}_2(S)$, by $[\Pi h](a) = E[h(S) \mid A = a]$, and for any $v$, let $\widehat\Pi^v:\mathcal{L}_2(S) \to \mathcal{L}_2(A)$, defined by $[\widehat\Pi^v h](a) = (n / L)^{-1} \sum_i h(S_i) \bm{1}\{A_i = a, V_i = v\}$. For any $f: \mathcal{O} \to \mathbb{R}$, and $u : \mathcal{A} \to \mathbb{R}$, let
\begin{align}
    \left\lVert f \right\rVert_{1, N} = \frac{1}{N} \sum_{i=1}^N | f(O_i)| \qquad \text{and} \qquad \left\lVert u \right\rVert_{1, K} = \frac{1}{K} \sum_{a=1}^K |u(a)|.
\end{align}
For any $f \in \mathcal{L}_2(S)$, let 
\begin{align}
    \rho_N(f) =  \left\lVert \Pi f \right\rVert_{1, K}  +  \left\lVert f \right\rVert_{1, N}.
\end{align}

\begin{assumption}[Bounded range]\label{asm:bounded_range}
    It holds that $Y \in [-1, 1]$ a.s. and that $\sup_{h \in \mathcal{H}} \left\lVert h \right\rVert_\infty \leq 1$.
\end{assumption}

\begin{assumption}[Entropy]\label{asm:entropy}
    For any $\epsilon > 0$, the $\rho_N$-covering entropy of $\mathcal{H}$ satisfies $\log N(\epsilon, \mathcal{H}, \rho_N) = o_P(N)$.
\end{assumption}

\begin{theorem}[Nonparametric consistency]\label{thm:nonparametric_consistency}
    Suppose that assumptions \ref{asm:bounded_range} and \ref{asm:entropy} hold. Then as $K \to \infty$, $n$ stays fixed, 
    \begin{align}
        R(\widehat h) - \inf_{h \in \mathcal{H}} R(h) = o_P(1).
    \end{align}
\end{theorem}

\section{Linear instantiation: $L$-fold JIVE}

\subsection{Estimator derivation under a linear structural model}

As a practical illustration of our method, we consider cross-fold risk minimizers under a linear structural model. Specifically, we suppose that we have $K$ cells with $n$ units each, and that for each unit $i$, the long-term outcome $Y_i$ and the short term metrics $S_i$ satisfy
\begin{align}
    Y_i = \underbrace{S_i}_{1 \times d} \beta + \delta \underbrace{U_i}_{1 \times 1} + \epsilon_i \qquad \text{and} \qquad S_i =\underbrace{A_i}_{1 \times K} \underbrace{\Pi}_{K \times d} + \gamma U_i + \eta_i, \label{eq:linear_structural_model}
\end{align}
where $A_i$ is the one-hot encoding of cell-membership (we abuse the notation by using interchangeably cell number and one-hot encoding), $E[U_i \mid A_i] = 0$, $E[\epsilon_i \mid A_i, S_i, U_i] = 0$, and $E[\eta_i \mid A_i, U_i] = 0$. Dividing each cell in $L$ equal folds, such that $V_i \in [L]$ is the fold membership of unit $i$, we can show that the minimizer of the empirical counterpart of the cross-fold risk \eqref{eq:cross-fold_risk} can be expressed only from cell-fold-level aggregates as follows:
\begin{align}
    \widehat \beta = \left(\sum_{a=1}^K \sum_{v = 1}^L {\bar{S}_{a,-v}}^\top \bar S_{a,v}\right)^{-1} \sum_{a=1}^K \sum_{v = 1}^L {\bar{S}_{a,-v}}^\top \bar Y_{a,v}, \qquad &\text{where} \qquad \bar S_{a,-v} = \frac{1}{n \frac{L-1}{L}} \sum_i S_i \bm{1}\{A_i = a, V_i \neq v\},\\
    \bar S_{a,v} = \frac{1}{n / L} \sum_i S_i \bm{1}\{A_i = a, V_i = v\} \qquad &\text{and} \qquad \bar Y_{a,v} = \frac{1}{n / L} \sum_i Y_i \bm{1}\{A_i = a, V_i = v\}.
\end{align}
The expression of $\widehat \beta$ above is a $L$-fold version of the Jackknife Instrumental Variable Estimator (JIVE) of \cite{angrist1999jackknife} (the original version of JIVE uses leave-one-out sample splitting as opposed of $L$-fold for fixed $L$ here). That $\widehat \beta$ can be expressed from fold-level aggregates is of great practical significance for organizations that conduct a large number of AB tests: fold-level aggregates can be easily logged for each experiment, which then renders computing $\widehat \beta$ quick and easy.

\subsection{Asymptotics in a simplified case}

The asymptotic analysis of the $L$-fold JIV estimator differs from the analysis of the original LOO JIVE. 
In this section, we provide guarantees for a simplified two-fold JIV estimator under a simplified setting with one-endogenous variable and homoskedastic iid errors (specifically, we assume that the pairs $(\eta_i, \epsilon_i)$ are i.i.d.).  We provide asymptotics under a sequence of DGPs indexed by $m = 1,2, \ldots$. For $v \in \{0,1\}$, we denote $\bar \eta_{a,v,m}$ and $\bar \epsilon_{a,v,m}$ the fold-level errors, and $\sigma^2_{\bar \eta,m} = E \bar \eta_{a,v,m}^2$ and $\sigma^2_{\bar \epsilon,m} = E \bar \epsilon_{a,v,m}^2$. We define the $m$-th two-fold JIVE as a function of a draw of the $m$-th DGP as follows:
\begin{align}
    \widehat \beta_m = \left( \sum_{a=1}^K \bar S_{a,0,m} \bar S_{a,1,m} \right)^{-1} \sum_{a=1}^K \bar S_{a,0,m} \bar Y_{a,1,m}.
\end{align}
We denote $\pi_{a,m}$ the first-stage effect under the $m$-th DGP in cell $a$. We suppose that the number of cells/instruments $K_m \to \infty$.
We make the following assumptions. 
\begin{assumption}[Instrument strength]\label{asm:instrument_strength}
    It holds that
    \begin{align}
        \frac{\sum_{a=1}^K \pi_{a,m}^2}{K_m \sigma_{\bar \eta, m}^2} \to 0, \qquad \frac{\sigma_{\bar \epsilon, m}}{\sigma_{\bar \eta, m}} \to c \in [0, \infty), \qquad \text{and} \qquad \frac{\sum_{a=1}^K \pi_{a,m}^2}{\sqrt{K_m} \sigma_{\bar \eta, m}^2} \to \infty.
    \end{align}
\end{assumption}
The first part of assumption \ref{asm:instrument_strength} defines the weak instrument setting. It says that the average strength of the instruments converges to zero. The second part is a technical condition. The third part is a common assumption in the JIVE literature stating that the ratio of the so-called concentration parameter (the quantity $\sum_{a=1}^{K_m} \pi_{a,m}^2 / \sigma_{\bar \eta, m}^2$) and of the square root of the number of instrument diverges. This condition imposes a lower bound on the average instrument strength.
Our next assumption is a classic Lindeberg condition on the errors.
\begin{assumption}[Lindeberg]\label{asm:lindeberg}
    It holds that 
    \begin{align}
        E\left[ \left( \frac{\bar \eta_{1,0,m} \bar \epsilon_{1,1,m}}{\sigma_{\bar \eta, m} \sigma_{\bar \epsilon, m}}  \right)^2 \bm{1} \left\lbrace \left\lvert \frac{\bar \eta_{1,0,m} \bar \epsilon_{1,1,m}}{\sigma_{\bar \eta, m} \sigma_{\bar \epsilon, m}} \right\rvert \geq \epsilon \sqrt{K_m}\right\rbrace \right] \to 0 \qquad \text{for every } \epsilon > 0. 
    \end{align}
\end{assumption}
We can now state the asymptotics theorem.
\begin{theorem}\label{thm:lin_JIVE_asymptotics}
    Let $\widehat{\sigma}_{\bar \eta, m}$ and $\widehat{\sigma}_{\bar \epsilon, m}$ be such that $\widehat{\sigma}_{\bar \eta, m} / \sigma_{\bar \eta, m} \xrightarrow{P} 1$ and $\widehat{\sigma}_{\bar \epsilon, m} / \sigma_{\bar \epsilon, m} \xrightarrow{P} 1$. Suppose assumptions \ref{asm:instrument_strength} and \ref{asm:lindeberg} hold. Let 
    \begin{align}
        \mathcal{C}_m = \left[ \widehat \beta_m \pm q_{1-\alpha/2} \left(\frac{\sum_{a=1}^K \bar S_{a,0,m} \bar S_{a,1,m}}{\sqrt{K_m}} \right)^{-1} \widehat \sigma_{\bar \eta,m} \widehat \sigma_{\bar \epsilon, m} \right],
    \end{align}
    where $q_{1-\alpha / 2}$ is the $1-\alpha/2$ quantile of the normal distribution. Then 
    $\Pr_m[\beta \in \mathcal{C}_m] \to 1 - \alpha$ and $\beta_m \xrightarrow{P} \beta$.
\end{theorem}

\section{Numerical experiments}

\begin{figure}[t!]
\centering
\begin{subfigure}[m]{0.475\textwidth}
\centering
\includegraphics[scale=0.5]{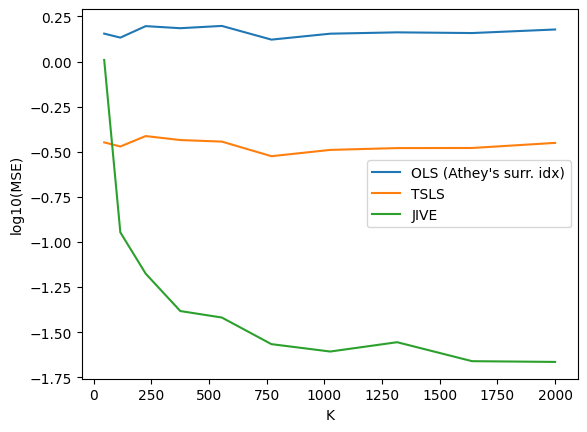}
\caption{MSE}\label{fig:MSE}
\end{subfigure}
\hfill
\begin{subfigure}[m]{0.475\textwidth}
\centering
\includegraphics[scale=0.5]{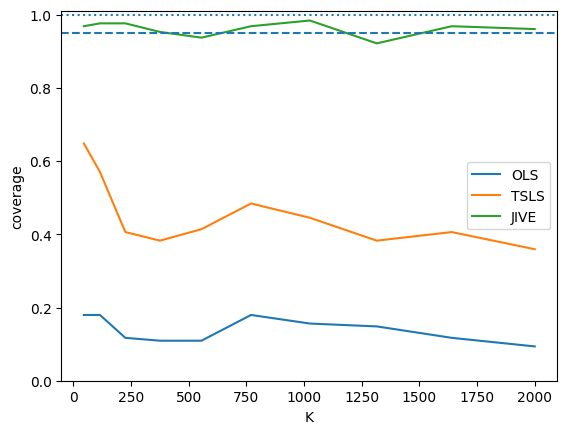}
\caption{Coverage}\label{fig:coverage}
\end{subfigure}
\end{figure}

We compare the performance of $5$-fold JIVE, two-stage least square (TSLS) and the linear surrogate-index method (OLS of $Y$ on $S$) in learning a surrogate bridge function of the form $h(S) = S \beta$. We simulate $K$ cells (for $K$ ranging from 45 to 2000) with $100$ units each where the data-generating process is of the form specified in \eqref{eq:linear_structural_model}. We set $d = 5$, $\epsilon_i \sim 3 \times \mathcal N(0,1)$, $\eta_i \sim \mathcal{N}(0, I_5)$, $U \sim 3 \times \mathcal{N}(0,1)$, $\delta = 1$. We draw the first-stage treatment effects, that is the rows of $\Pi$ effects $0.1 \times \mathcal{N}(0, I_5)$. We draw $\beta \sim \mathcal{N}(0, I_5) / \sqrt{5}$, and $\gamma \sim \mathcal{N}(0, I_5) / \sqrt{5}$. In the novel cell $a'$, we set the vector $a' \Pi$ of first-stage treatment effects to $(1,1,1,1,1)$.

We estimate the mean squared error $E[(a'\Pi \beta - S(a') \widehat \beta )^2]$ of the learned surrogate w.r.t. the true long-term outcome $E[Y(a')] = a'\Pi \beta$ in the novel cell $a'$ and the coverage of $E[Y(a')]$ by confidence intervals (accounting for the randomness in $\widehat \beta$ coming from historical data sampling and that in $S(a')$ coming from novel-cell data sampling).

\bibliographystyle{abbrvnat}
\bibliography{lit}

\appendix

\section{Proof of the identification result (theorem \ref{thm:identification})}

The proof of theorem \ref{thm:identification} relies on the following intermediate results.

\begin{lemma}\label{lemma:S2u2indepU2S1}
Suppose that assumptions \ref{asm:potential_outcomes}-\ref{asm:randomization}
hold. Then, $(U_2, S_1) \indep S_2(u_2)$ for every $u_2$.
\end{lemma}

\begin{proof}
    \begin{align}
        &\pr[S_2(u_2) = s_2 \mid S_1 = s_1, U_2 = u_2] \\
        =& \sum_a \pr[S_2(u_2) = s_2 \mid S_1(a) = s_1, U_2(a) = u_2, A= a] \\
        & \qquad \times \pr[A=a \mid S_1 = s_1, U_2 = u_2]\\
        =& \pr[S_2(u_2) = s_2] \sum_a \pr[A=a \mid S_1 = s_1, U_2 = u_2] \\
        =& \pr[S_2(u_2) = s_2],
    \end{align} 
    where the second equality follows from assumptions \ref{asm:indep_pot_outcomes} and \ref{asm:randomization}.
\end{proof}

\begin{lemma}\label{lemma:U2S1blockAtoS2}
    Suppose that assumptions \ref{asm:potential_outcomes}-\ref{asm:randomization}
hold. Then $S_2 \indep A \mid U_2, S_1$.
\end{lemma}

\begin{proof}
    \begin{align}
        &\pr[S_2  = s_2 \mid U_2 = u_2, S_1 = s_1, A = a] \\
        =& \pr[S_2(u_2) = s_2 \mid U_2(a) = u_2, S_1(a) = s_1, A = a] \\
        =& \pr[S_2(u_2) = s_2] \\
        =& \pr[S_2(u_2) = s_2 \mid U_2 = u_2, S_1 = s_1]\\
        =& \pr[S_2 = s_2 \mid U_2 = u_2, S_1 = s_1],
    \end{align}
    where the first equality follows from assumption \ref{asm:potential_outcomes}, the second one from assumptions \ref{asm:indep_pot_outcomes} and \ref{asm:randomization}, the third one from lemma \ref{lemma:S2u2indepU2S1}, and the last one from assumption \ref{asm:potential_outcomes}.
\end{proof}

\begin{lemma}\label{lemma:unique_bridge} Suppose that assumptions \ref{asm:potential_outcomes}-\ref{asm:randomization} and \ref{asm:bridge_existence}-\ref{asm:completeness}
hold. Then, if $E[Y \mid A] = E[f(S_1, S_2, U_2) \mid A]$ almost surely for some function $f$, then $((s_1, u_2) \mapsto E[f(S_1, S_2, U_2) - h(S_1, S_2) \mid S_1 = s_1, U_2]) = 0$.
\end{lemma}

\begin{proof} From assumption \ref{asm:bridge_existence}, 
    \begin{align}
        0 =& E[f(S_1, S_2, U_2) - h(S_1, S_2) \mid A ] \\
        =& E[ E[f(S_1, S_2, U_2) - h(S_1, S_2) \mid S_1, U_2] \mid A ],
    \end{align}
    where the second equality follows from lemma \ref{lemma:U2S1blockAtoS2}. Therefore, from assumption \ref{asm:completeness}, $E[f(S_1, S_2, U_2) - h(S_1, S_2) \mid S_1=s_1, U_2=u_2]$ for every $s_1, u_2$.
\end{proof}

We can now prove our main identification result.
\begin{proof}[Proof of theorem \ref{thm:identification}]
Let $f(s_1, s_2, u_2) = E[Y(s_1, s_2, u_2)]$. We have that
    \begin{align}
        &E[Y(a)] \\
        =& \sum_{s_1, s_2, u_1, u_2} E[Y(s_1, s_2, u_2) \mid S_1(a)= s_1, S_2(u_2) = s_2, U_1=u_1, U_2(a) = u_2] \\
        & \qquad \times \pr[S_1(a) = s_1, S_2(u_2) = s_2, U_1 = u_1, U_2(a) = U_2] \\
        =& \sum_{s_1, s_2, u_1, u_2} E[Y(s_1, s_2, u_2)] \times \pr[S_1(a) = s_1, S_2(u_2) = s_2, U_2(a) = U_2] \\ 
        & + \sum_{s_1, s_2, u_1, u_2} (E[Y(s_1, s_2, u_2)\mid U_1=u_1] - E[Y(s_1, s_2, u_2)])\\
        & \qquad \times \pr[S_1(a) = s_1, S_2(u_2) = s_2, U_1 = u_1, U_2(a) = U_2] \\
        =& \sum_{s_1, s_2, u_1, u_2} E[Y(s_1, s_2, u_2)] \times \pr[S_1(a) = s_1, S_2(u_2) = s_2, U_2(a) = U_2] \\
        =& E[f(S_1, S_2, U_2) \mid A = a] \\
        =& E[ E[f(S_1, S_2, U_2) \mid S_1, U_2] \mid A = a] \\
        =& E[ E[h(S_1, S_2) \mid S_1, U_2] \mid A = a] \\
        =& E[h(S_1, S_2) \mid A = a],
    \end{align}
    where 
    \begin{itemize}
        \item the first equality follows from assumptions \ref{asm:randomization} and \ref{asm:potential_outcomes},
        \item the third equality follows from assumption \ref{asm:homogeneous_treatment_effect},
        \item the fifth equality follows from the fact that $S_2 \indep A \mid S_1, U_2$ from lemma \ref{lemma:U2S1blockAtoS2},
        \item the sixth equality follows from lemma \ref{lemma:unique_bridge},
        \item the last equality follows again from lemma \ref{lemma:U2S1blockAtoS2}.
    \end{itemize}
\end{proof}

\section{Proof of the nonparametric consistency result (theorem \ref{thm:nonparametric_consistency})}

\begin{lemma}[Discretization error]\label{lemma:discretization_error}
Suppose assumption \ref{asm:bounded_range} holds. Then 
\begin{align}
    \sup_{\rho(h_1 - h_2) \leq \epsilon} (R(h_1) - \widehat{R}(h_1)) - (R(h_2) - \widehat{R}(h_2)) \leq 2 \epsilon.
\end{align}
\end{lemma}

\begin{lemma}[Expected finite maximum]\label{lemma:expected_finite_maximum}
Suppose that assumptions \ref{asm:bounded_range}. Then, for any $h_1,\ldots,h_m \in \mathcal{H}$, 
\begin{align}
    E \left[ \max_{j \in [m]} \left\lvert R(h_j) - \widehat R(h_j) \right\rvert \right] \lesssim \sqrt{\frac{\log(1 + m)}{N}}.
\end{align}
\end{lemma}

\begin{proof}[Proof of lemma \ref{lemma:discretization_error}]
    Let $\widetilde y: [\mathcal{O} \to \mathbb{R}] : o = (s,a,y) \mapsto y$. We have that
    \begin{align}
        &R(h_2) - \widehat R(h_2) - (R(h_1) - \widehat R(h_1)) =\\
        & \frac{1}{2} \left\lbrace \frac{1}{K} \sum_{a=1}^K [\Pi(h_2 - h_1)](a) [\Pi h_2](a)  + \frac{1}{K} \sum_{a=1}^K [\Pi h_1](a) [\Pi (h_2 - h_1)](a) \right\rbrace\\
        & - \frac{1}{2}\left\lbrace \frac{1}{L}  \sum_{v=1}^L \frac{1}{N/L} \sum_{i:V_i = v}^N (h_2 - h_1)(S_i) [\widehat \Pi^{-v} h_2](A_i) + \frac{1}{L} \sum_{v=1} \frac{1}{N \frac{L-1}{L}} \sum_{i : V_i \neq v} [\widehat \Pi^v h_1](A_i) (h_2-h_1)(A_i) \right\rbrace \\
        & - \left\lbrace \frac{1}{K}  \sum_{a=1}^K [\Pi(h_2 - h_1)](a) [\Pi \widetilde y](a) - \frac{1}{L} \sum_{v=1} \frac{1}{N \frac{L-1}{L}} \sum_{i : V_i \neq v} (h_2 - h_1)(S_i) [\widehat \Pi^v \widetilde y](A_i) \right\rbrace.
    \end{align}
    Therefore, from the triangle inequality and assumption \ref{asm:bounded_range}, if $\rho_N(h_2 - h_1) \leq \epsilon$, 
    \begin{align}
        & \left\lvert R(h_2) - \widehat R(h_2) - (R(h_1) - \widehat R(h_1)) \right\rvert \\
        \leq & \frac{1}{2} \left\lbrace \left\lVert \Pi(h_2 - h_1) \right\rVert_{1,K}  +\left\lVert \Pi(h_2 - h_1) \right\rVert_{1,K} \right\rbrace \\
        &+ \frac{1}{2} \left\lbrace  \left\lVert h_2 - h_1\right\rVert_{1,N} + \left\lVert h_2 - h_1\right\rVert_{1,N} \right\rbrace \\
        & + \left\lbrace \left\lVert \Pi(h_2 - h_1) \right\rVert_{1,K}  + \left\lVert h_2 - h_1\right\rVert_{1,N}  \right\rbrace \\
        =& 2 \rho_N(h_2 - h_1).
    \end{align}
\end{proof}

\begin{proof}[Proof of lemma \ref{lemma:expected_finite_maximum}]
    We have that
    \begin{align}
         R(h) - \widehat R(h) = & \frac{1}{L} \sum_{v=1}^L \frac{1}{K}\sum_{a=1}^K [(\widehat{\Pi}^{-v} - \Pi) h](a) [\widehat{\Pi}^{v} \widetilde y](a) \\
         & + \frac{1}{L} \sum_{v=1}^L \frac{1}{K}\sum_{a=1}^K [ \Pi h](a) [(\widehat{\Pi}^{v} - \Pi) \widetilde y](a) \\
         &+ \frac{1}{2 L} \sum_{v=1}^L \frac{1}{K}\sum_{a=1}^K [(\widehat{\Pi}^{-v} - \Pi) h](a) [\widehat{\Pi}^{v} h](a)\\
         &+ \frac{1}{2 L} \sum_{v=1}^L \frac{1}{K}\sum_{a=1}^K [\Pi h](a) [(\widehat{\Pi}^{v} - \Pi) h](a).
    \end{align}
    Therefore, $E[\max_{j\in [m]} | R(h_j) - \widehat R(h_j)| ]$ can be bounded, up to a universal constant, by a constant (depending on $V$) number of terms of the form
    \begin{align}
        E \left[ \max_{j \in [N]} \frac{1}{N} \sum_{i \in \mathcal{B}} (h_j(S_i) - E[h_j(S_i) \mid A_i = a]) B_{A_i} \right], 
    \end{align}
    where $\mathcal{B}$ is either $\{i : V_i = v \}$ or $\{i : V_i \neq v \}$ for some $v \in [L]$, and $B_1,\ldots,B_K$ are $\{O_i : i \in [N] \backslash \mathcal{B} \}$-measurable random variables ranging in $[-1,1]$.
    
\begin{align}
    &E \left[ \max_{j \in [N]} \frac{1}{N} \sum_{i \in \mathcal{B}} (h_j(S_i) - E[h_j(S_i) \mid A_i = a]) B_{A_i} \right] \\
    =&
    E \left[E \left[ \max_{j \in [N]} \frac{1}{N} \sum_{i \in \mathcal{B}} (h_j(S_i) - E[h_j(S_i) \mid A_i = a]) B_{A_i}  \mid O_i : i \not\in \mathcal{B}, A_i \right]\right] \\
    \leq & 2\frac{1}{\sqrt{N}} E \left[E \left[ \max_{j \in [N]} \frac{1}{\sqrt{N}} \sum_{i \in \mathcal{B}} \epsilon_i h_j(S_i) B_{A_i}  \mid O_i : i \not\in \mathcal{B}, A_i \right]\right] \\
    \leq & 2 \left\lVert  \max_{j \in [N]} \frac{1}{\sqrt{N}} \sum_{i \in \mathcal{B}} \epsilon_i h_j(S_i) B_{A_i} \right\rVert_{\Psi_2 \mid O_i : i \in [N] } \\
    \lesssim &  \sqrt{\frac{\log(1 + m)}{N}} \max_{j \in [m]} \sqrt{\frac{1}{N} \sum_{i \in \mathcal{B}} (h_j(S_i) B_{A_i})^2 } \\
    \lesssim & \sqrt{\frac{\log(1 + m)}{N}},
\end{align}
where the third line follows from a symmetrization argument in which $\epsilon_1, \epsilon_2, \ldots$ are i.i.d. Rademacher random variables independent of $O_1,\ldots,O_N$, where in the fourth line $\left\lVert \cdot \right\rVert_{\Psi_2 \mid O_i : i \in [N]}$ is the Orlicz norm conditional on $O_1,\ldots,O_N$, associated to $\Psi_2(x) = \exp(x^2) - 1$, and where the before last line follows from (i) a classical result on the Orclicz norms of a finite maximum, and (ii) from another classical result on the $\Psi_2$-Orlicz norm of a Rademacher average.
\end{proof}

\begin{proof}[Proof of theorem \ref{thm:nonparametric_consistency}]
    Let $\epsilon > 0$, denote $m = N(\epsilon, \mathcal{H}, \rho_N)$, and let $h_1,\ldots, h_m$ be a an $\epsilon$-covering in $\rho_N$ norm of $\mathcal{H}$.
    From the definition of $\widehat h$, we have that
    \begin{align}
        R(\widehat{h}) - \inf_{h \in \mathcal{H}} R(h) \leq & \sup_{h \in \mathcal{H}} R(h) - \widehat{R}(h) \\
        \leq & \sup_{\rho(h-h') \leq \epsilon} \left\lvert (R(h) - \widehat R(h)) - (R(h') - \widehat R(h')) \right\rvert \\
        & + \max_{j \in [m]}  R(h_j) - \widehat R(h_j).
    \end{align}
    The first term above is bounded by $\epsilon$ from assumption \ref{asm:bounded_range}
via lemma \ref{lemma:discretization_error}, and the $L_1$-norm of the second term is bounded by a constant times $\sqrt{\log(1+N(\epsilon, \mathcal{H}, \rho_N)) / (K n)}$. Since convergence in $L_1$ norm implies convergence in probability, taking $\epsilon$
 and $1 / K$ arbitrarily small yields the claim.
 \end{proof}

\section{Proof of theorem \ref{thm:lin_JIVE_asymptotics}}
\begin{proof}
    Denoting $\widehat H_K = \sum_{a=1}^K \bar S_{a,0} \bar{S}_{a,1}$, we have that
    \begin{align}
        &\Pr[\beta \in \mathcal{C}_m] \to 1 - \alpha \\
        \iff & \Pr \left[ \frac{\widehat H_K}{\sqrt{K}}\frac{\widehat \beta - \beta}{\widehat \sigma_{\bar{\eta}} \widehat \sigma_{\bar \epsilon}} \in [\pm q_{1-\alpha / 2}]\right] \to 1-\alpha \\
        \iff & \Pr \left[N_K  \in [\pm q_{1-\alpha / 2}] \right]  \to 1-\alpha,
    \end{align}
    where 
    \begin{align}
        N_K =& N_{1,K} + N_{2,K} \\
        \text{with} \qquad N_{1,K} =& \frac{1}{\sqrt{K}} \sum_{a=1}^K \frac{\pi_a \bar{\epsilon}_{a,1}}{\sigma_{\bar{\eta}} \sigma_{\bar{\epsilon}} }  \qquad \text{and} \qquad N_{1,K} = \frac{1}{\sqrt{K}} \sum_{a=1}^K \frac{\bar{\eta}_{a,0} \bar{\epsilon}_{a,1}}{\sigma_{\bar{\eta}} \sigma_{\bar{\epsilon}} }.
    \end{align}
\end{proof}
We have $N_{1,K} = O_P(\sum_{a=1}^K \pi_a^2 / (K \sigma_{\bar \eta}^2)) = o_P(1)$ from the first part of assumption \ref{asm:instrument_strength}. 
From assumption \ref{asm:lindeberg}, Lindeberg's central limit theorem yields that $N_{2,K} \xrightarrow{d} \mathcal{N}(0,1)$, and thus $N_K \xrightarrow{d} \mathcal{N}(0,1)$, which yields the coverage claim.

For the consistency claim to hold, it suffices that $\widehat H_K / \sqrt{K \sigma_{\bar \eta}^2 \sigma_{\bar \epsilon}^2} \xrightarrow{P} \infty$. We have that
\begin{align}
    \widehat H_K =& H_{1,K} + H_{2,K} + H_{3,K},\\
    \text{where} \qquad  H_{1,K} =& \sum_{a=1}^K \pi_a^2,\\
    H_{2,K} =& \sum_{a=1}^K \pi_a (\bar{\eta}_{a,0} + \bar{\eta}_{a,1}) \\
    H_{3,K} =& \sum_{a=1}^K \bar{\eta}_{a,0} \bar{\eta}_{a,1}.
\end{align}
We have that 
\begin{align}
    \frac{H_{2,K}}{H_{1,K}} =& O_P \left( \frac{\sqrt{\sum_{a=1}^K \pi_a^2 \sigma_{\bar \eta}^2}}{\sum_{a=1}^K \pi_a^2}\right) = O_P \left( K^{-1/4} \sqrt{\frac{\sqrt{K} \sigma_{\bar \eta}^2}{\sum_{a=1}^K \pi_a^2}} \right)= o_P(1),
\end{align}
where the last equality follows from the third part of assumption \ref{asm:instrument_strength}, and 
\begin{align}
    \frac{H_{3,K}}{H_{1,K}} = O_P \left( \frac{\sigma_{\bar \eta} \sqrt{K}}{\sum_{a=1}^K \pi_a^2} \right) = o_P(1),
\end{align}
where the last equality also follows from the third part of assumption \ref{asm:instrument_strength}.
Therefore, using the second and the third part of assumption \ref{asm:instrument_strength}, we have that
\begin{align}
    \frac{\widehat H_K}{\sqrt{K} \widehat \sigma_{\bar \eta} \widehat \sigma_{\bar \epsilon}} = (1 + o_P(1)) \frac{\sum_{a=1}^K \pi_a^2}{ \sqrt{K} \sigma_{\bar \eta}^2} \frac{\sigma_{\bar{\eta}}}{\sigma_{\bar{\epsilon}}} \to \infty,
\end{align}
which yields the consistency claim.
\end{document}